\documentclass[12pt]{iopart}%
\usepackage{amssymb}
\usepackage{amsfonts}%
\usepackage{graphicx}
\newtheorem{theorem}{Theorem}

\newtheorem{definition}[theorem]{Definition}

\newtheorem{lemma}[theorem]{Lemma}

\newtheorem{proposition}[theorem]{Proposition}

\newenvironment{proof}[1][Proof]{\noindent\textbf{#1.} }{\ \rule{0.5em}{0.5em}}
\begin{document}

\title[Lower bound for the distance between classical and quantum spin correlations]{Lower bound for the mean square distance between classical and quantum spin correlations}
\author{G Gr\"{u}bl and L Wurzer}
\address{Theoretical Physics Institute, Innsbruck University,\\Technikerstr. 25, A-6020 Innsbruck, Austria}
\ead{\mailto{gebhard.gruebl@uibk.ac.at}, \mailto{lukas.wurzer@uibk.ac.at}}

\begin{abstract}
Bell's theorem prevents local Kolmogorov-simulations of the singlet state of
two spin-1/2 particles. We derive a positive lower bound for the $L^{2}%
$-distance between the quantum mechanical spin singlet anticorrelation
function $\cos$ and any of its classical approximants $C$ formed by the
statio\-nary autocorrelation functions of mean-square-continuous, $2\pi
$-periodic, $\pm1$-valued, stochastic processes. This bound is given by
$\left\Vert C-\cos\right\Vert \geq\left(  1-\frac{8}{\pi^{2}}\right)
/\sqrt{2}\approx0.133\,95.$

\end{abstract}

\section{Introduction}

Consider a Stern-Gerlach correlation experiment performed on the two distant
members of a spin-singlet state of two spin 1/2 particles. The two
apparatuses' magnetic fields point into the directions $a,$ and $b.$ Quantum
mechanics tells that the stochastic correlation $-Q\left(  a,b\right)  $
between the two measurements' results $\pm1$ is given by $-a\cdot b,$ the
negative scalar product between the two analyzing directions' unit vectors $a
$ and $b.$ To get rid of the minus sign we employ in the following the
anticorrelation function $Q.$

Bell demonstrated that any local Kolmogorov simulation of such experiments is
at variance with $Q\left(  a,b\right)  =a\cdot b.$ \cite{Bel64} The basic
reason for this being his inequality, which subjects any triple of values
$\mathcal{C}\left(  a,b\right)  ,\mathcal{C}\left(  a,c\right)  ,$ and
$\mathcal{C}\left(  b,c\right)  $ of the autocorrelation function
$\mathcal{C}:\mathbb{S}^{2}\times\mathbb{S}^{2}\rightarrow\left[  -1,1\right]
,\left(  a,b\right)  \mapsto\left\langle f_{a}f_{b}\right\rangle $ of any
local Kolmogorov simulation to%
\[
\left\vert \mathcal{C}\left(  a,b\right)  -\mathcal{C}\left(  a,c\right)
\right\vert \leq1-\mathcal{C}\left(  b,c\right)  .
\]
The quantum mechanical function $Q,$ however, does not obey Bell's inequality
for all choices of directions $a,b,c.$ Thus $Q=\mathcal{C}$ cannot hold.

This raises the question for how close the quantum mechanical function $Q$ can
be approximated by a local Kolmogorov autocorrelation function $\mathcal{C}.$
Clearly a notion of \textquotedblleft closeness\textquotedblright\ is needed
in order to make the question meaningful. The mean square deviation between
$\mathcal{C}$ and $Q$ averaged over a suitable set of pairs of directions is a
natural choice, which we adopt. How should one choose the set of directions to
be averaged over? Rotation invariance suggests to keep one point
$a\in\mathbb{S}^{2}$ fixed and to average over all $b$ from a great circle or
equivalently a half great circle through $a,$ because from the $SO\left(
3\right)  $ invariance\footnote{This means that $\mathcal{C}\left(
Ra,Rb\right)  =\mathcal{C}\left(  a,b\right)  $ holds for all directions $a,b$
and for all $R\in SO\left(  3\right)  .$} of $\mathcal{C}$ it follows that
$\mathcal{C}$ is completely determined by the mapping $b\mapsto\mathcal{C}%
\left(  a,b\right)  $ for any fixed $a$ and all directions $b$ lying on a
meridian emanating from $a.$ After all $\mathcal{C}$ is represented by a
function of the (unoriented) angle between $a$ and $b$ only, i.e., there
exists a function $C:\left[  0,\pi\right]  \rightarrow\left[  -1,1\right]  $
such that $\mathcal{C}\left(  a,b\right)  =C\left(  \theta_{a,b}\right)  $ for
all $a,b\in\mathbb{S}^{2}.$ Here $\theta_{a,b}\in\left[  0,\pi\right]  $ is
uniquely determined by $a\cdot b=\cos\theta_{a,b}.$

Within the class of local Kolmogorov autocorrelation functions which are
rotation invariant and continuous we have found that the mean square deviation
of $C$ from its quantum counterpart $\cos$ obeys the estimate%
\[
\sqrt{\frac{1}{\pi}\int_{0}^{\pi}\left\vert C\left(  \theta\right)
-\cos\theta\right\vert ^{2}\rmd\theta}\geq\left(  1-\frac{8}{\pi^{2}}\right)
/\sqrt{2}\approx0.133\,95.
\]

The search for such \textquotedblleft integrated\textquotedblright\ conditions
which constrain the local Kolmogorov approximants $C$ of the quantum
anticorrelation function $\cos$ seems to have been initiated by \.{Z}ukowski,
who in \cite{Zuk93} derived an upper bound for the modulus of the $L^{2}%
$-scalar product between $\cos$ and $C.$ See estimate (10) in \cite{Zuk93},
where it was taken for granted that certain integrability conditions are
obeyed by the stochastic variables of the local Kolmogorov simulation under
consideration. In \cite{KaZ00} the idea has been generalized further in order
to cover more general, entangled many particle states.

In a first step we prove a scalar product bound, which in certain cases is
equivalent to \.{Z}ukowski's one. Our bound only holds for the singlet state
but it does not rely on \.{Z}ukowski's integrability assumption. The class of
Kolmogorov simulations for which it holds is specified precisely in our lemma
5. From this we then derive in proposition 6 a lower bound for the mean square
distance between $C$ and $\cos.$ Although fairly straightforward this latter
bound seems to have remained unnoticed till now. It could serve the following purpose.

Usually the decision upon the possibility of a local Kolmogorov simulation of
a given set of spin singlett correlation measurements typically is based on
the degree of violation of the Clausner-Horne-Shimony-Holt inequality for a
finite set of analyzing field directions. \cite{AGR81} Our inequality might
give a more sensitive criterion as to whether an experimentally determined
singlet correlation function admits a local Kolmogorov simulation, since it
involves the sum of (squared) deviations at a sufficient dense set of
analyzing directions. If the mean square distance of the experimental data
points from the quantum mechanically predicted values is estimated within a
certain confidence level to be less than $0.13395,$ then a local Kolmogorov
simulation can be ruled out within the same level of confidence. And this may
well be the case even if all the individual vioaltions of the CHSH-inequality
do not rule out locality with sufficient confidence.

\section{Local Kolmogorov-simulation}

Let $\left(  \chi_{+},\chi_{-}\right)  $ denote an orthonormal basis of
$\mathbb{C}^{2}.$ Then the vector%
\[
\chi=\frac{1}{\sqrt{2}}\left(  \chi_{+}\otimes\chi_{-}-\chi_{-}\otimes\chi
_{+}\right)
\]
represents the spin-singlet state of two spin-1/2 particles. A Stern-Gerlach
experiment on each of the two constituents of the singlet state, which are
assumed to be approximately localized in well separated regions, measures the
two observables $A=\sigma\left(  a\right)  \otimes id$ and $B=id\otimes
\sigma\left(  b\right)  $ where $a,b\in\mathbb{S}^{2}\subset\mathbb{R}^{3}$
are two arbitrary vectors from the unit sphere $\mathbb{S}^{2}.$ The unit
vectors $a,b$ specify the Stern-Gerlach apparatuses' orientations. For
$a=\left(  a^{1},a^{2},a^{3}\right)  \in\mathbb{R}^{3}$ holds%
\[
\sigma\left(  a\right)  =\left(
\begin{array}
[c]{cc}%
a^{3} & a^{1}-\rmi a^{2}\\
a^{1}+\rmi a^{2} & -a^{3}%
\end{array}
\right)  .
\]

The possible results of such measurements are the pairs $\left(
\varepsilon,\eta\right)  \in\left\{  1,-1\right\}  \times\left\{
1,-1\right\}  .$ The quantum mechanically determined probability of outcome
$\left(  \varepsilon,\eta\right)  $ is given in terms of the euclidean scalar
product $a\cdot b$ between $a$ and $b$ through%
\begin{equation}
p_{a,b}\left(  \varepsilon,\eta\right)  =\frac{1-\varepsilon\eta\ a\cdot b}%
{4}.\label{Qprob}%
\end{equation}
This probability does not depend on the time order of the two Stern-Gerlach
experiments since $\left[  A,B\right]  =0.$

Thus any choice of two directions $a,b$ determines a probability measure on
the event space $\left\{  1,-1\right\}  \times\left\{  1,-1\right\}  .$ The
situation is analogous to a random experiment in which a pair of widely
separated coins is tossed, while each coin of the pair is exposed to a
magnetic field of direction $a$ and $b$ respectively. And the magnetic fields
take an influence on the distribution of the possible results.

The family of probability functions $\left\{  p_{a,b}:a,b\in\mathbb{S}%
^{2}\right\}  $ can be obtained as the distributions of stochastic variables
$\left\{  X_{a,b}:a,b\in\mathbb{S}^{2}\right\}  $%
\[
X_{a,b}=\left(  f_{a,b},g_{a,b}\right)  :\Omega\rightarrow\left\{
1,-1\right\}  \times\left\{  1,-1\right\}
\]
on a \emph{single} Kolmogorov probability space $\left(  \Omega,W\right)  $ as
follows. \cite{Wig70} Take as the space of events $\Omega$ the square $\left[
0,1\right]  \times\left[  0,1\right]  $ with the uniform distribution $W.$
Then decompose the square $\Omega$ into a set of four nonoverlapping
rectangles, i.e.
\[
\Omega=\bigcup\nolimits_{\varepsilon,\eta\in\left\{  1,-1\right\}  }%
R_{a,b}\left(  \varepsilon,\eta\right)  ,
\]
such that the area of the rectangle $R_{a,b}\left(  \varepsilon,\eta\right)  $
has the value $p_{a,b}\left(  \varepsilon,\eta\right)  .$ Finally assume%
\[
X_{a,b}\left(  \omega\right)  =\left(  \varepsilon,\eta\right)  \mbox{ for 
}\omega\in R_{a,b}\left(  \varepsilon,\eta\right)  .
\]
Then, by construction we obviously have%
\[
p_{a,b}\left(  \varepsilon,\eta\right)  =W\left(  \left\{  \omega\in
\Omega:X_{a,b}\left(  \omega\right)  =\left(  \varepsilon,\eta\right)
\right\}  \right)  .
\]
In such a model the outcome of each coin flip is determined by a randomly
chosen \textquotedblleft hidden variable\textquotedblright\ $\omega\in\Omega$
in a way that in general the outcome also depends on both of the magnetic
fields $a$ and $b!$ How come that a coin gets influenced by faraway circumstances?

Trying to save locality, Bell considered the question whether there exists a
probability space $\left(  \Omega,W\right)  $ with a stochastic variable
\[
X_{a,b}=\left(  f_{a},g_{b}\right)  :\Omega\rightarrow\left\{  1,-1\right\}
\times\left\{  1,-1\right\}
\]
for each pair $\left(  a,b\right)  \in\mathbb{S}^{2}\times\mathbb{S}^{2}$ such
that%
\begin{equation}
p_{a,b}\left(  \varepsilon,\eta\right)  =W\left(  \left\{  \omega\in
\Omega:f_{a}\left(  \omega\right)  =\varepsilon\mbox{ and }g_{b}\left(
\omega\right)  =\eta\right\}  \right)  \mbox{ for all }a,b\in\mathbb{S}
^{2}.\label{Kprob}%
\end{equation}

Such structure would constitute a \emph{local Kolmogorov simulation} of the
probability distributions generated by composite Stern-Gerlach experiments on
a spin singlet state. The assumption that the stochastic variables
$f_{a},g_{b}$ depend on their respective Stern-Gerlach orientation $a$ or $b$
only is made in order to take care of the principle of locality.

Now the equations (\ref{Qprob}) and (\ref{Kprob}) imply
\[
\left\langle f_{a}g_{b}\right\rangle =\sum_{\varepsilon,\eta\in\left\{
1,-1\right\}  }\varepsilon\eta p_{a,b}\left(  \varepsilon,\eta\right)
=-a\cdot b
\]
for all $a,b\in\mathbb{S}^{2}.$ Yet Bell's theorem \cite{Bel64} rules out
exactly that $\left\langle f_{a}g_{b}\right\rangle =-a\cdot b$ for all
$a,b\in\mathbb{S}^{2}.$ Therefore a local Kolmogorov simulation of the singlet
state does not exist. For the sake of completeness we spell out Bell's theorem precisely.

\begin{theorem}
[Bell]Let $\left(  \Omega,W\right)  $ be a probability space with two stochastic
variables $f_{a},g_{a}:\Omega\rightarrow\left\{  1,-1\right\}  $ for every
$a\in\mathbb{S}^{2}.$ Then there exist points $a,b\in\mathbb{S}^{2}$ such that $\left\langle f_{a}g_{b}\right\rangle \neq-a\cdot b.$

\end{theorem}

\begin{proof}
That $\left\langle f_{a}g_{b}\right\rangle =-a\cdot b$ cannot hold for all
$a,b\in\mathbb{S}^{2}$ may be proven by contradiction. Choosing $b=a$ the
equation $\left\langle f_{a}g_{a}\right\rangle =-1$ implies $g_{a}=-f_{a}$ in
the sense of stochastic variables, i.e., almost everywhere on $\Omega.$ Thus,
assuming $\left\langle f_{a}g_{b}\right\rangle =-a\cdot b$ for all
$a,b\in\mathbb{S}^{2}$ leads to%
\[
-\left\langle f_{a}g_{b}\right\rangle =\left\langle f_{a}f_{b}\right\rangle
=a\cdot b
\]
for all $a,b\in\mathbb{S}^{2}.$ Bell noticed that, due to $f_{b}^{2}=1,$ there
holds%
\[
\left\vert \left\langle f_{a}f_{b}\right\rangle -\left\langle f_{a}%
f_{c}\right\rangle \right\vert =\left\vert \left\langle f_{a}f_{b}\left(
1-f_{b}f_{c}\right)  \right\rangle \right\vert \leq1-\left\langle f_{b}%
f_{c}\right\rangle
\]
for all $a,b,c\in\mathbb{S}^{2}.$ Bell's famous inequality, however, is in
contradiction with $\left\langle f_{a}f_{b}\right\rangle =a\cdot b.$ Choose,
e.g., three coplanar vectors $a,b,c$ with $a\cdot b=\frac{1}{\sqrt{2}}=b\cdot
c$ and $a\cdot c=0.$ Then
\[
\frac{1}{\sqrt{2}}=\left\vert a\cdot b-a\cdot c\right\vert \leq1-b\cdot
c=\frac{\sqrt{2}-1}{\sqrt{2}}
\]
and thus the falsity $2\leq\sqrt{2}$ follows from Bell's inequality.
\end{proof}

\section{Quality of a classical singlet model}

Bell's theorem poses the following problem: Determine the infimum of the set
of numbers%
\[
\frac{1}{\left(  4\pi\right)  ^{2}}\int_{\mathbb{S}^{2}\times\mathbb{S}^{2}%
}\left\vert \left\langle f_{a}g_{b}\right\rangle -\left[  -a\cdot b\right]
\right\vert ^{2}\rmd a\ \rmd b,
\]
obtained from all \emph{local} Kolmogorov models of a spin singlet state. Any
such model\footnote{We denote such models as \emph{classical singlet models}.}
consists of a probability space $\left(  \Omega,W\right)  $ and two families
of $\left\{  1,-1\right\}  $-valued stochastic variables $\left\{  f_{a}%
:a\in\mathbb{S}^{2}\right\}  $ and $\left\{  g_{a}:a\in\mathbb{S}^{2}\right\}
$ such that $\left\langle f_{a}g_{a}\right\rangle =-1$ for all $a\in
\mathbb{S}^{2}.$ Thus again $g_{a}=-f_{a}$ holds for all $a\in\mathbb{S}^{2}.$
Here $\emph{da}$ and $\emph{db}$ denote the rotation invariant area element on
the unit sphere normalized to $4\pi.$

The solution to this problem would quantify and limit the optimal
approximation to the quantum mechanical covariance function
\[
\mathbb{S}^{2}\times\mathbb{S}^{2}\ni\left(  a,b\right)  \mapsto\left\langle
\chi,\sigma\left(  a\right)  \otimes\sigma\left(  b\right)  \chi\right\rangle
=-a\cdot b
\]
through classical singlet models.

In this paper we address a somewhat simpler but related problem. We first
confine the admissible direction vectors $a,b$ from $\mathbb{S}^{2}$ to a
great circle $\mathbb{S}^{1}\subset\mathbb{S}^{2}.$ Then we restrict to such
$\mathbb{S}^{1}$-parametrized families of stochastic variables $f_{a}=-g_{a}$
for all $a\in\mathbb{S}^{1},$ for which there exists a \emph{continuous}
function $\widetilde{C}:\mathbb{R}\rightarrow\mathbb{R}$ such that%
\[
\left\langle f_{a}f_{b}\right\rangle =\widetilde{C}\left(  a\cdot b\right)
\mbox{ holds for all }a,b\in\mathbb{S}^{1}.
\]

The assumption that $\left\langle f_{a}f_{b}\right\rangle $ depends on the
scalar product $a\cdot b$ only amounts to postulating $O\left(  2\right)
$-invariance for the mapping $\left(  a,b\right)  \mapsto\left\langle
f_{a}f_{b}\right\rangle ,$ i.e., the relation $\left\langle f_{Ra}%
f_{Rb}\right\rangle =\left\langle f_{a}f_{b}\right\rangle $ for each
orthogonal mapping $R$ which stabilizes the great circle $\mathbb{S}^{1}.$
Under these premises we shall derive a positive lower bound for%
\[
\frac{1}{\left(  2\pi\right)  ^{2}}\int_{\mathbb{S}^{1}\times\mathbb{S}^{1}%
}\left\vert \left\langle f_{a}f_{b}\right\rangle -a\cdot b\right\vert
^{2}\rmd a\ \rmd b.
\]
Here $\emph{da}$ and $\emph{db}$ denote the rotation invariant line element on
the unit circle normalized to $2\pi.$

An equivalent but simpler formulation of our problem is obtained by
periodically parametrizing the circle $\mathbb{S}^{1}$ through real numbers
$s$ and $t,$ e.g., such that $a\left(  s\right)  =\left(  \cos\left(
s\right)  ,\sin\left(  s\right)  ,0\right)  $ and $b\left(  s\right)  =\left(
\cos\left(  t\right)  ,\sin\left(  t\right)  ,0\right)  .$ Then there exists a
continuous $2\pi$-periodic function $C:\mathbb{R}\rightarrow\mathbb{R}$ with
$\widetilde{C}\left(  a\cdot b\right)  =C\left(  t-s\right)  $ and we obtain%
\begin{eqnarray*}                                                  
& \frac{1}{\left(  2\pi\right)  ^{2}}\int_{\mathbb{S}^{1}\times\mathbb{S}^{1}%
}\left\vert \left\langle f_{a}f_{b}\right\rangle -a\cdot b\right\vert
^{2}\rmd a\ \rmd b\\
& =\frac{1}{\left(  2\pi\right)  ^{2}}\int_{0}^{2\pi}\int_{0}^{2\pi}\left\vert
C\left(  t-s\right)  -\cos\left(  t-s\right)  \right\vert ^{2}\rmd s\ \rmd t\\
& =\frac{1}{2\pi}\int_{0}^{2\pi}\left\vert C\left(  t\right)  -\cos\left(
t\right)  \right\vert ^{2}\rmd t.
\end{eqnarray*}

Thus we try to find a positive lower bound for%
\[
\frac{1}{2\pi}\int_{0}^{2\pi}\left\vert C\left(  t\right)  -\cos\left(
t\right)  \right\vert ^{2}\rmd t,
\]
where the continuous function $C:\mathbb{R}\rightarrow\mathbb{R}$ is related
to a $2\pi$-periodic, $\left\{  1,-1\right\}  $-valued, stochastic process
$\left\{  f_{s}:s\in\mathbb{R}\right\}  $ through $C\left(  t\right)
=\left\langle f_{s}f_{s+t}\right\rangle $ for all $s,t\in\mathbb{R}.$

Notice that, because of%
\[
C\left(  -t\right)  =\left\langle f_{s}f_{s-t}\right\rangle =\left\langle
f_{s-t}f_{s-t+t}\right\rangle =\left\langle f_{s}f_{s+t}\right\rangle
=C\left(  t\right)  ,
\]
the function $C$ is even and it also obeys $C\left(  0\right)  =1.$

If, in addition to $C\left(  t\right)  =\left\langle f_{s}f_{s+t}\right\rangle
,$ the mapping $s\mapsto\left\langle f_{s}\right\rangle $ is constant, the
process $\left\{  f_{s}:s\in\mathbb{R}\right\}  $ is called stationary in the
wide sense. \cite{Fis63}, \cite{Tod92} In case of $\left\langle f_{s}%
\right\rangle =0$ for all $s\in R,$ the function $C$ specializes to the
autocorrelation function of a wide-sense-stationary process $\left\{
f_{s}:s\in\mathbb{R}\right\}  .$ However we will not need any assumption on
the expectation values $\left\langle f_{s}\right\rangle .$

The following results from the theory of stationary processes make it clear
that insistence on the continuity of $C$ is much less a restriction than it
appears to be. They show that the condition of continuity of $C$ can be
replaced by the seemingly weaker condition that $C$ is continuous at $0.$ See,
e.g., sect. 8.10 from \cite{Fis63}. For the sake of completeness we include
the proofs.

\begin{lemma}
Let $\left\{  f_{s}:s\in\mathbb{R}\right\}  $ be a $\left\{  1,-1\right\}
$-valued stochastic pro\-cess such that there exists a function $C:\mathbb{R}%
\rightarrow\mathbb{R}$ with $C\left(  t\right)  =\left\langle f_{s}%
f_{s+t}\right\rangle $ for all $s,t\in\mathbb{R}.$ Then $C$ is continuous
everywhere if and only if it is continuous at $0.$
\end{lemma}

\begin{proof}
Observe first that%
\[
C\left(  t+\varepsilon\right)  -C\left(  t\right)  =\left\langle
f_{0}f_{t+\varepsilon}\right\rangle -\left\langle f_{0}f_{t}\right\rangle
=\left\langle f_{0}\left(  f_{t+\varepsilon}-f_{t}\right)  \right\rangle .
\]
Thus we have, due to the Cauchy-Schwarz inequality and due to $f_{t}^{2}=1,$
that%
\begin{eqnarray*}
\left\vert C\left(  t+\varepsilon\right)  -C\left(  t\right)  \right\vert
^{2}  & =\left\vert \left\langle f_{0}\left(  f_{t+\varepsilon}-f_{t}\right)
\right\rangle \right\vert ^{2}\leq\left\langle f_{0}^{2}\right\rangle
\left\langle \left(  f_{t+\varepsilon}-f_{t}\right)  ^{2}\right\rangle \\
& =\left\langle \left(  f_{t+\varepsilon}-f_{t}\right)  ^{2}\right\rangle
=\left\langle f_{t+\varepsilon}^{2}+f_{t}^{2}-2f_{t}f_{t+\varepsilon
}\right\rangle \\
& =2\left(  1-C\left(  \varepsilon\right)  \right)  =2\left(  C\left(
0\right)  -C\left(  \varepsilon\right)  \right)  .
\end{eqnarray*}
From this it follows that%
\[
\lim_{\varepsilon\rightarrow0}\left\vert C\left(  t+\varepsilon\right)
-C\left(  t\right)  \right\vert \leq\sqrt{2\lim_{\varepsilon\rightarrow
0}\left(  C\left(  0\right)  -C\left(  \varepsilon\right)  \right)  }.
\]
Thus $\lim_{\varepsilon\rightarrow0}\left\vert C\left(  t+\varepsilon\right)
-C\left(  t\right)  \right\vert =0$ if $\lim_{\varepsilon\rightarrow0}C\left(
\varepsilon\right)  =C\left(  0\right)  .$ Clearly, if $C$ is continuous
everywhere it is in particular continuous at $0.$
\end{proof}

\begin{definition}
A stochastic process $\left\{  f_{s}:s\in\mathbb{R}\right\}  ,$ with
$\lim_{\varepsilon\rightarrow0}\left\langle \left(  f_{t+\varepsilon}%
-f_{t}\right)  ^{2}\right\rangle =0$ for all $t\in\mathbb{R}$ is called mean-square-continuous.
\end{definition}

\begin{lemma}
A stochastic process $\left\{  f_{s}:s\in\mathbb{R}\right\}  $ with $C\left(
t\right)  =\left\langle f_{s}f_{s+t}\right\rangle $ for all $s,t\in\mathbb{R}
$ with values in $\left\{  1,-1\right\}  $ is mean-square-continuous if and
only if $C$ is continuous at $0.$
\end{lemma}

\begin{proof}
In the case of a $\left\{  1,-1\right\}  $-valued process with stationary
correlation function $C$ holds%
\begin{eqnarray*}
\lim_{\varepsilon\rightarrow0}\left\langle \left(  f_{t+\varepsilon}%
-f_{t}\right)  ^{2}\right\rangle  & =\lim_{\varepsilon\rightarrow0}\left[
\left\langle f_{t+\varepsilon}^{2}\right\rangle +\left\langle f_{t}%
^{2}\right\rangle -2\left\langle f_{t+\varepsilon}f_{t}\right\rangle \right]
=2\lim_{\varepsilon\rightarrow0}\left(  1-C\left(  \varepsilon\right)  \right)
\\
& =2\lim_{\varepsilon\rightarrow0}\left(  C\left(  0\right)  -C\left(
\varepsilon\right)  \right)  .
\end{eqnarray*}

\end{proof}

Thus in the case of a $\left\{  1,-1\right\}  $-valued stochastic process
$\left\{  f_{s}:s\in\mathbb{R}\right\}  ,$ for which there exists a function
$C:\mathbb{R}\rightarrow\mathbb{R}$ such that $C\left(  t\right)
=\left\langle f_{s}f_{s+t}\right\rangle $ holds for all $s,t\in\mathbb{R},$
the following three conditions are equivalent:

\begin{enumerate}
\item $C$ is continuous.

\item $C$ is continuous at $0.$

\item $\left\{  f_{s}:s\in\mathbb{R}\right\}  $ is mean-square-continuous.
\end{enumerate}

\section{Lower bound for the quality}

Our main tool is the following estimate for the Fourier coefficients of the
correlation functions $C$ which appear in the present context.

\begin{lemma}
Let $\left\{  f_{s}:s\in\mathbb{R}\right\}  $ be a $2\pi$-periodic, $\left\{
1,-1\right\}  $-valued, mean-square-con\-tinuous stochastic process such that
$C\left(  t\right)  =\left\langle f_{s}f_{s+t}\right\rangle $ holds for all
$s,t\in\mathbb{R}.$ Then all the Fourier coefficients $c_{k}$ of $C$ exist and
for all $k\in\mathbb{Z}$ the following estimates hold%
\[
0\leq c_{k}=\frac{1}{2\pi}\int_{0}^{2\pi}\rme^{-ikt}C\left(  t\right)
\rmd t=\frac{1}{2\pi}\int_{0}^{2\pi}\cos\left(  kt\right)  C\left(
t\right)  \rmd t\leq\left(  \frac{2}{\pi}\right)  ^{2}.
\]

\end{lemma}

\begin{proof}
Since the function $C$ is continuous, its Fourier coefficients exist. Since
$C$ is real valued and even, there holds $\overline{c_{k}}=c_{-k}=c_{k}.$ Thus
the mapping $k\mapsto c_{k}$ is real valued and even too. In particular
because of $c_{k}\in\mathbb{R}$ we have%
\begin{eqnarray*}
\int_{0}^{2\pi}\rme^{-\rmi kt}C\left(  t\right)  \rmd t  & =\int_{-\pi
}^{\pi}\rme^{-\rmi kt}C\left(  t\right)  \rmd t=\int_{-\pi}^{\pi}\left(
\cos\left(  kt\right)  -\rmi\sin\left(  kt\right)  \right)  C\left(  t\right)
\rmd t\\
& =\int_{-\pi}^{\pi}\cos\left(  kt\right)  C\left(  t\right)  \rmd t%
=\int_{0}^{2\pi}\cos\left(  kt\right)  C\left(  t\right)  \rmd t.
\end{eqnarray*}
Now for the estimate $c_{k}\geq0:$%
\begin{eqnarray*}
c_{k}  & =\frac{1}{2\pi}\int_{0}^{2\pi}\rme^{-\rmi kt}\left\langle f_{s}%
f_{s+t}\right\rangle \rmd t=\frac{1}{2\pi}\int_{0}^{2\pi}\rme^{-\rmi kt}\left(
\frac{1}{2\pi}\int_{0}^{2\pi}\left\langle f_{s}f_{s+t}\right\rangle
\rmd s\right)  \rmd t\\
& =\left(  \frac{1}{2\pi}\right)  ^{2}\int_{0}^{2\pi}\rme^{\rmi ks}\left(  \int
_{0}^{2\pi}\rme^{-\rmi k\left(  s+t\right)  }\left\langle f_{s}f_{s+t}\right\rangle
\rmd t\right)  \rmd s\\
& =\left(  \frac{1}{2\pi}\right)  ^{2}\int_{0}^{2\pi}\rme^{\rmi ks}\left(  \int
_{0}^{2\pi}\rme^{-\rmi kr}\left\langle f_{s}f_{r}\right\rangle \rmd r\right)
\rmd s\\
& =\lim_{N\rightarrow\infty}\sum_{m=1}^{N}\sum_{n=1}^{N}\frac{\exp\left(\rmi k2\pi
\frac{m}{N}\right)\exp\left(-\rmi k2\pi\frac{n}{N}\right)}{N^{2}}\left\langle f_{2\pi\frac{m}{N}%
}f_{2\pi\frac{n}{N}}\right\rangle \\
& =\lim_{N\rightarrow\infty}\left\langle \left\vert \sum_{n=1}^{N}%
\frac{\exp\left(-\rmi k2\pi\frac{n}{N}\right)}{N}f_{2\pi\frac{n}{N}}\right\vert ^{2}%
\right\rangle \geq0.
\end{eqnarray*}

The last line's sum can be bounded from above as follows. First observe that
for all $\tau\in\mathbb{R}$ it holds%
\[
\left\vert \sum_{n=1}^{N}\frac{\exp\left(-\rmi k2\pi\frac{n}{N}\right)}{N}f_{2\pi\frac{n}{N}%
}\right\vert =\left\vert \rme^{\rmi k\tau}\sum_{n=1}^{N}\frac{\exp\left(-\rmi k2\pi\frac{n}{N}\right)%
}{N}f_{2\pi\frac{n}{N}}\right\vert .
\]
Now there exists a number $\tau_{N}\in\left[  0,\frac{2\pi}{k}\right)  $
depending on $N,$ such that%
\[
0\leq \rme^{\rmi k\tau_{N}}\sum_{n=1}^{N}\frac{\exp\left(-\rmi k2\pi\frac{n}{N}\right)}{N}f_{2\pi
\frac{n}{N}}=\left\vert \rme^{\rmi k\tau_{N}}\sum_{n=1}^{N}\frac{\exp\left(-\rmi k2\pi\frac
{n}{N}\right)}{N}f_{2\pi\frac{n}{N}}\right\vert
\]
holds. It then follows for such $\tau_{N}$ that%
\[
\left\vert \sum_{n=1}^{N}\frac{\exp\left(-\rmi k2\pi\frac{n}{N}\right)}{N}f_{2\pi\frac{n}{N}%
}\right\vert =\sum_{n=1}^{N}\frac{\cos\left(  k\left(  2\pi\frac{n}{N}%
-\tau_{N}\right)  \right)  }{N}f_{2\pi\frac{n}{N}}.
\]
Since $f_{2\pi\frac{n}{N}} $ assumes values from $\left\{  -1,1\right\}$  only, we obtain
from the triangle inequality%
\[
\sum_{n=1}^{N}\frac{\cos\left(  k\left(  2\pi\frac{n}{N}-\tau_{N}\right)
\right)  }{N}f_{2\pi\frac{n}{N}}\leq\frac{1}{2\pi}\sum_{n=1}^{N}\frac
{2\pi\left\vert \cos\left(  k\left(  2\pi\frac{n}{N}-\tau_{N}\right)  \right)
\right\vert }{N}.
\]

The points $\mathcal{Z}_{N}=\left\{  2\pi\frac{n}{N}-\tau_{N}:n=1,\ldots
N\right\}  $ partition the interval $\left[  0,2\pi\right]  -\tau_{N}$ whose
length is $2\pi.$ The mesh $2\pi/N$ of $\mathcal{Z}_{N}$ tends to zero for
$N\rightarrow\infty.$ Furthermore $2\pi$ is a periode of the function
$x\mapsto\left\vert \cos\left(  kx\right)  \right\vert .$ Therefore the sum on
the right hand side of this last inequality converges towards the Riemannian
integral%
\[
\frac{1}{2\pi}\int_{0}^{2\pi}\left\vert \cos\left(  kt\right)  \right\vert
\rmd t
\]
for $N\rightarrow\infty.$ For each $\varepsilon>0$ there thus exists a number
$N_{0},$ such that for all $N>N_{0}$%
\[
\sum_{n=1}^{N}\frac{\left\vert \cos\left(  k2\pi\frac{n}{N}-\tau_{N}\right)
\right\vert }{N}\leq\frac{1}{2\pi}\int_{0}^{2\pi}\left\vert \cos\left(
kt\right)  \right\vert \rmd t+\varepsilon.
\]

For $k\in\mathbb{Z}\smallsetminus0$ we have%
\[
\int_{0}^{2\pi}\left\vert \cos\left(  kt\right)  \right\vert \rmd t%
=4k\int_{0}^{\frac{\pi}{2k}}\cos\left(  kt\right)  \rmd t=4\left.
\sin\left(  kt\right)  \right\vert _{0}^{\pi/2k}=4.
\]
Thus for any $\varepsilon>0$ there exists a $N_{0},$ such that for all
$N>N_{0}$ holds%
\[
\left\vert \sum_{n=1}^{N}\frac{\exp\left(-\rmi k2\pi\frac{n}{N}\right)}{N}f_{2\pi\frac{n}{N}%
}\right\vert \leq\frac{2}{\pi}+\varepsilon.
\]
We therefore have proven for any $k\in\mathbb{Z}\smallsetminus0$ that%
\[
c_{k}=\lim_{N\rightarrow\infty}\left\langle \left\vert \sum_{n=1}^{N}%
\frac{\exp\left(-\rmi k2\pi\frac{n}{N}\right)}{N}f_{2\pi\frac{n}{N}}\right\vert ^{2}%
\right\rangle \leq\left(  \frac{2}{\pi}\right)  ^{2}.
\]
For $k=0$ the estimate $c_{0}\leq\left(  \frac{2}{\pi}\right)  ^{2}$ follows
from
\[
c_{0}=\frac{1}{2\pi}\int_{0}^{2\pi}C\left(  t\right)  \rmd t=0.
\]

\end{proof}

Note that in this proof we did not interchange the limiting process of
integration with the probabilistic expectation value, which in general also
involves a limit process. Such an interchange can be misleading since the
realizations $t\mapsto f_{t}\left(  \omega\right)  $ need not be integrable
for almost all $\omega\in\Omega.$ If, however, the two limits can be
interchanged, the proof gets abbreviated considerably.\cite{Zuk93},
\cite{Wur10}

From the lemma's estimate for the case $k=1$ namely $0\leq c_{1}\leq\left(
\frac{2}{\pi}\right)  ^{2}$ we now obtain our lower bound for the $L^{2}%
$-distance between the quantum mechanical spin singlet correlation function
and its classical approximants.

\begin{proposition}
Let $\left\{  f_{s}:s\in\mathbb{R}\right\}  $ be a $2\pi$-periodic, $\left\{
1,-1\right\}  $-valued, mean-square-con\-tinuous stochastic process such that
$C\left(  t\right)  =\left\langle f_{s}f_{s+t}\right\rangle $ holds for all
$s,t\in\mathbb{R}.$ Then the mean square deviation of $C$ from the quantum
mechanical correlation function $\cos$ obeys%
\[
\left\Vert C-\cos\right\Vert =\sqrt{\frac{1}{2\pi}\int_{0}^{2\pi}\left(
C\left(  t\right)  -\cos\left(  t\right)  \right)  ^{2}\rmd t}\geq
\frac{1-\frac{8}{\pi^{2}}}{\sqrt{2}}\approx0.133\,95.
\]

\end{proposition}

\begin{proof}
Note that%
\[
\left\Vert \cos\right\Vert ^{2}=\frac{1}{2\pi}\int_{0}^{2\pi}\cos^{2}\left(
t\right)  \rmd t=\frac{1}{2}.
\]
We thus can decompose $C$ into a component parallel to $\cos$ and one
orthogonal to it according to%
\[
C=\sqrt{2}\cos\left\langle \sqrt{2}\cos,C\right\rangle +\left(  C-\sqrt{2}%
\cos\left\langle \sqrt{2}\cos,C\right\rangle \right)  .
\]
Here the scalar product between two continuous functions $f,g:\mathbb{R}%
\rightarrow\mathbb{C}$ with a periode of $2\pi$ is denoted by%
\[
\left\langle f,g\right\rangle =\frac{1}{2\pi}\int_{0}^{2\pi}\overline{f\left(
t\right)  }g\left(  t\right)  \rmd t.
\]
From the estimate
\[
c_{1}=\frac{1}{2\pi}\int_{0}^{2\pi}\cos\left(  t\right)  C\left(  t\right)
\rmd t\leq\frac{4}{\pi^{2}}
\]
it follows that
\begin{eqnarray*}
\left\Vert C-\cos\right\Vert ^{2}  & =\frac{1}{2\pi}\int_{0}^{2\pi}\left(
C\left(  t\right)  -\cos\left(  t\right)  \right)  ^{2}\rmd t\\
& =\left\Vert C-\sqrt{2}\cos\left\langle \sqrt{2}\cos,C\right\rangle
\right\Vert ^{2}+\left\Vert \sqrt{2}\cos\left\langle \sqrt{2}\cos
,C\right\rangle -\cos\right\Vert ^{2}\\
& \geq\left\Vert \sqrt{2}\cos\left\langle \sqrt{2}\cos,C\right\rangle
-\cos\right\Vert ^{2}\\
& =\left(  1-2\left\langle \cos,C\right\rangle \right)  ^{2}\left\Vert
\cos\right\Vert ^{2}\\
& =\frac{1}{2}\left(  1-2\left\langle \cos,C\right\rangle \right)  ^{2}.
\end{eqnarray*}
Because of $0\leq2\left\langle \cos,C\right\rangle =2c_{1}\leq\frac{8}{\pi
^{2}}=0.810\,57$ we finally have%
\[
\left\Vert C-\cos\right\Vert \geq\frac{1-\frac{8}{\pi^{2}}}{\sqrt{2}%
}=0.133\,95.
\]

\end{proof}

From this proof it is obvious that the estimate $\left\Vert C-\cos\right\Vert
\geq\left(  1-\frac{8}{\pi^{2}}\right)  /\sqrt{2}$ is saturated if and only if
$C$ is proportional to $\cos$ which, because of $C\left(  0\right)  =1,$ in
turn implies $C=\cos.$ Thus the estimate cannot be saturated and stronger
estimates might exist.

\section{Bell's example}

Bell \cite{Bel64} constructed a local spin singlet model with the $2\pi
$-periodic autocorrelation function given by $C:\mathbb{R}\rightarrow\left[
-1,1\right]  $:%
\[
C\left(  t\right)  =1-2\frac{\left\vert t\right\vert }{\pi}\mbox{ for }
0\leq\left\vert t\right\vert \leq\pi.
\]
$C$ is continuous and even. Bell's stochastic variables $\left\{  f_{a}%
:a\in\mathbb{S}^{2}\right\}  $ are defined on the set $\mathbb{S}^{2}$ endowed
with the uniform distribution. They are given through%
\[
f_{a}\left(  \omega\right)  =\left\{
\begin{array}
[c]{rl}%
1 & \mbox{for }\omega\cdot a>0\\
-1 & \mbox{otherwise}
\end{array}
\right\}  ,
\]
and indeed yield $\left\langle f_{a}f_{b}\right\rangle =1-2\frac{\theta}{\pi}$
with $\theta\in\left[  0,\pi\right]  $ such that $a\cdot b=\cos\theta$ holds.

The Fourier coefficient $c_{k}$ of $C$ is given for $k\in\mathbb{Z}%
\smallsetminus0$ through%
\begin{eqnarray*} 
c_{k}  & =\frac{1}{2\pi}\int_{-\pi}^{\pi}\rme^{-\rmi kt}\left(  1-2\frac{\left\vert
t\right\vert }{\pi}\right)  \rmd t=\frac{1}{2\pi}\int_{-\pi}^{\pi}%
\cos\left(  kt\right)  \left(  1-2\frac{\left\vert t\right\vert }{\pi}\right)
\rmd t\\
& =\frac{1}{\pi}\int_{0}^{\pi}\cos\left(  kt\right)  \left(  1-2\frac{t}{\pi
}\right)  \rmd t=\frac{1}{k\pi}\int_{0}^{\pi}\frac{d}{dt}\left[
\sin\left(  kt\right)  \right]  \left(  1-2\frac{t}{\pi}\right)  \rmd t\\
& =\frac{1}{k\pi}\left\{  \left[  \sin\left(  kt\right)  \left(  1-2\frac
{t}{\pi}\right)  \right]  _{0}^{\pi}+\frac{2}{\pi}\int_{0}^{\pi}\sin\left(
kt\right)  \rmd t\right\} \\
& =-\frac{2}{\left(  k\pi\right)  ^{2}}\left.  \cos\left(  kt\right)
\right\vert _{0}^{\pi}=\frac{2}{\left(  k\pi\right)  ^{2}}\left[  1-\left(
-1\right)  ^{k}\right] \\
& =\left\{
\begin{array}
[c]{cl}%
\frac{4}{\left(  k\pi\right)  ^{2}} & \mbox{for odd }k\\
0 & \mbox{for even }k
\end{array}
\right\}  .
\end{eqnarray*}
Obviously $c_{0}=0$ holds. The mapping $k\mapsto c_{k}$ indeed is real valued
and even. The estimate $0\leq c_{k}\leq4/\pi^{2}$ is realized and in the case
of $c_{1}$ saturated. Therefore it holds that%
\begin{eqnarray*}
\left\Vert C-\cos\right\Vert ^{2}  & =\left\Vert C\right\Vert ^{2}+\left\Vert
\cos\right\Vert ^{2}-2\left\langle \cos,C\right\rangle \\
& =\left\Vert C\right\Vert ^{2}+\frac{1}{2}-2c_{1}\\
& =\left\Vert C\right\Vert ^{2}+\frac{1}{2}-\frac{8}{\pi^{2}}.
\end{eqnarray*}

The value of $\left\Vert C\right\Vert ^{2}$ is given by%
\begin{eqnarray*}
\left\Vert C\right\Vert ^{2}  & =\frac{1}{2\pi}\int_{-\pi}^{\pi}C^{2}\left(
t\right)  \rmd t=\frac{1}{\pi}\int_{0}^{\pi}C^{2}\left(  t\right)
\rmd t\\
& =\frac{1}{\pi}\int_{0}^{\pi}\left(  1-2\frac{t}{\pi}\right)  ^{2}%
\rmd t=\int_{0}^{1}\left(  1-2x\right)  ^{2}\rmd x\\
& =\frac{1}{2}\int_{-1}^{1}y^{2}\rmd y=\int_{0}^{1}y^{2}\rmd y=\frac
{1}{3}.
\end{eqnarray*}
Thus we have%
\[
\left\Vert C-\cos\right\Vert ^{2}=\frac{1}{3}+\frac{1}{2}-\frac{8}{\pi^{2}%
}=\frac{5}{6}-\frac{8}{\pi^{2}},
\]
and in consequence%
\[
\left\Vert C-\cos\right\Vert =\sqrt{\frac{5}{6}-\frac{8}{\pi^{2}}}=0.150\,88.
\]

Note that $C$ has the following particularly simple uniformly converging
Fourier series representation%
\begin{eqnarray*}
C\left(  t\right)   & =\sum_{k=1}^{\infty}\left(  c_{k}\rme^{ikt}+c_{-k}%
\rme^{-ikt}\right)  =\sum_{k=1}^{\infty}c_{k}\left(  \rme^{ikt}+\rme^{-ikt}\right) \\
& =\sum_{k=1}^{\infty}2c_{k}\cos\left(  kt\right)  =\sum_{k=0}^{\infty
}2c_{2k+1}\cos\left(  \left(  2k+1\right)  t\right) \\
& =\frac{8}{\pi^{2}}\sum_{k=0}^{\infty}\frac{\cos\left(  \left(  2k+1\right)
t\right)  }{\left(  2k+1\right)  ^{2}}\\
& =\frac{8}{\pi^{2}}\left[  \cos\left(  t\right)  +\frac{\cos\left(
3t\right)  }{3^{2}}+\frac{\cos\left(  5t\right)  }{5^{2}}+\ldots\right]  .
\end{eqnarray*}

\ack

We are indebted to Gregor Weihs for advising us of reference \cite{Zuk93}. Critical remarks by Markus Penz and Tobias Griesser on an earlier version of the manuscript have been helpful.

\end{document}